\documentclass[a4paper,11pt]{article}

\usepackage{amssymb,amsmath,amsfonts,amstext}

\usepackage{url,bbm}

\usepackage{fullpage}
\usepackage{times}

\newtheorem{theorem}{Theorem}[section]
\newtheorem{lemma}[theorem]{Lemma}

\newcommand{\sq}{\hbox{\rlap{$\sqcap$}$\sqcup$}}
\newcommand{\qed}{\hspace*{\fill}\sq}
\newenvironment{proof}{\noindent {\bf Proof.}\ }{\qed\par\vskip 4mm\par}

\newcommand{\one}[1]{\mathbbm{1}#1}
\newcommand{\E}[1]{\mathbb{E}\left[ #1 \right]}
\newcommand{\bb}{\mathbf{b}}
\newcommand{\F}{\mathbf{F}}
\newcommand{\vv}{\mathbf{v}}
\newcommand{\R}{\mathbb{R}}
\newcommand{\cc}{\mathbf{c}}

\title{{\bf Welfare guarantees for proportional allocations}\thanks{This work was partially supported by the EC-funded STREP EULER and by the European Social Fund and Greek national funds through the research funding program Thales on ``Algorithmic Game Theory''.}}

\author{Ioannis Caragiannis\thanks{Computer Technology Institute and Press ``Diophantus'' \& Department of Computer Engineering and Informatics, University of Patras, 26504 Rion, Greece. Email: {\tt caragian@ceid.upatras.gr}}  \and  Alexandros A. Voudouris\thanks{Department of Computer Engineering and Informatics, University of Patras, 26504 Rion, Greece. Email: {\tt voudouris@ceid.upatras.gr}}}
\begin{document}

\maketitle

\begin{abstract}
According to the proportional allocation mechanism from the network optimization literature, users compete for a divisible resource -- such as bandwidth -- by submitting bids. The mechanism allocates to each user a fraction of the resource that is proportional to her bid and collects an amount equal to her bid as payment. Since users act as utility-maximizers, this naturally defines a proportional allocation game. Recently, Syrgkanis and Tardos (STOC 2013) quantified the inefficiency of equilibria in this game with respect to the social welfare and presented a lower bound of $26.8\%$ on the price of anarchy over coarse-correlated and Bayes-Nash equilibria in the full and incomplete information settings, respectively. In this paper, we improve this bound to $50\%$ over both equilibrium concepts. Our analysis is simpler and, furthermore, we argue that it cannot be improved by arguments that do not take the equilibrium structure into account. We also extend it to settings with budget constraints where we show the first constant bound (between $36\%$ and $50\%$) on the price of anarchy of the corresponding game with respect to an effective welfare benchmark that takes budgets into account.

\medskip\noindent{\bf Keywords: } algorithmic game theory, price of anarchy, Bayes-Nash equilibrium, proportional allocation mechanism
\end{abstract}

\section{Introduction}\label{sec:intro}
The {\em proportional allocation mechanism}, introduced by Kelly~\cite{K97}, is fundamental in the network optimization literature. According to this mechanism, a divisible resource --- such as bandwidth of a communication link --- is allocated to users as follows. Each user submits a bid to the mechanism; this corresponds to the user's {\em willingness-to-pay} for sharing the resource. The mechanism allocates to each user a fraction of the resource that is equal to the ratio of her bid over the total amount of bids. It also receives a payment from each user that is equal to her bid. This naturally defines a {\em proportional allocation game} among the users who act as players; each player has a (typically concave, non-negative, and non-decreasing) valuation function for the resource share she receives and aims to maximize her utility, i.e., her value for the resource share minus her payment to the mechanism. As it is typically the case in games, the {\em social welfare} (i.e., the total value of the players for the resource share they receive) at {\em equilibria} is, in general, suboptimal.

We aim to quantify this inefficiency of equilibria by bounding the {\em price of anarchy} \cite{KP99} of proportional allocation games. Besides the well-known work of Johari and Tsitsiklis \cite{JT04} who considered pure Nash equilibria in the full information setting, there has been surprisingly little focus on price of anarchy bounds over more general equilibrium concepts. The only exception we are aware of is the recent work of Syrgkanis and Tardos \cite{ST13} who studied proportional allocation as part of a broader class of mechanisms. Motivated by their work, we present new bounds on the price of anarchy of proportional allocation under general equilibrium concepts, such as {\em coarse-correlated} equilibria in the full information setting and {\em Bayes-Nash} equilibria in the incomplete information setting. In particular, we prove that the social welfare at equilibria is at least $1/2$ of the optimal social welfare. The bound holds for coarse-correlated and pure Bayes-Nash equilibria in the full information and Bayesian setting, respectively, and improves the bound of $26.8\%$ of \cite{ST13}. The proof is conceptually simple and is obtained by bounding the utility of every player at equilibrium by the utility this player would have by deviating to a particular deterministic bid.

We also consider the scenario where players have budget constraints representing their {\em ability-to-pay}. Here, each player has a budget and is never allowed to bid above it. We assess the quality of equilibria in this case in terms of an {\em effective welfare} benchmark --- proposed in previous work but further refined here --- that takes budgets into account. We show that the effective welfare at equilibrium is at least a constant fraction of the optimal one. To the best of our knowledge, this is the first constant price of anarchy bound (in particular, between $36\%$ and $50\%$) with respect to this benchmark\footnote{Previously, Syrgkanis and Tardos \cite{ST13} had shown that the social welfare at equilibria is at least $2-\sqrt{3}\approx 26.8\%$ of the optimal effective welfare. Our techniques can be used to improve this particular guarantee to $1/2$.}. Again, our proofs follow by considering a single deterministic deviation for each player, defined in a slightly different way compared to the deviation we consider in our bound on the social welfare.

\paragraph{Related work.} The proportional allocation mechanism and its variations have received significant attention in the network optimization literature. Proportional allocation games have been considered in \cite{HG02,LA00,MB03} where the existence and uniqueness conditions for pure Nash equilibria are proved. Variations of the mechanism with different definitions for the allocation rule or the payments have been considered in \cite{MB04,NT07,NV11,SH04} (see also the discussion in \cite{J07}).

Johari and Tsitsiklis \cite{JT04} were the first who assessed the quality of proportional allocations in terms of the social welfare. They focused on pure Nash equilibria and proved a lower bound of $3/4$ on their price of anarchy. Their analysis is based on the important observation that a pure Nash equilibrium in a proportional allocation game is also a pure Nash equilibrium in a game where each player has a {\em linear} valuation function with slope equal to the derivative of the original valuation function at the share value they get at equilibrium. The optimal social welfare in the new game is not smaller than the original one and this allows them to consider the significantly simpler case of linear valuations in their analysis. Then, the price of anarchy bound is obtained by solving a linear program. An alternative proof to the result of \cite{JT04} without using this argument is presented in \cite{R06} (see also \cite{J07}).

Unfortunately, this transformation does not apply to more general equilibrium concepts since the resource share each player receives is, in general, a random variable. This is a rather common difficulty that manifests itself in the analysis of games, as we depart from pure Nash equilibria and full information. In particular, Bayes-Nash equilibria have such an extremely rich structure that, typically, the price of anarchy analysis assesses their quality by rather ignoring this structure. Instead, it resorts to bounding the utility of each player by appropriately selected deviations which reveal a relation between the social welfare at equilibrium and the optimal social welfare. This approach has been used in a series of papers that mostly focus on auctions (e.g., see \cite{BR10,CKK+12,CKS08,FFGL13,KMST13,R12,ST13}) and is actually the approach we follow in the current paper as well.

Syrgkanis and Tardos \cite{ST13} present a general analysis framework for the broad class of smooth mechanisms. Among other results, they show a price of anarchy lower bound of $26.8\%$ over coarse-correlated and mixed Bayes-Nash equilibria of proportional allocation games. In their analysis, they bound the utility of each player by the utility she would have by deviating to an appropriately defined {\em randomized} bid (an approach that has also been used in different contexts in \cite{CKK+12,KMST13,S12,ST12}) with a probability distribution that depends only on the optimal allocation and the valuation function of the player. In contrast, the deviating bid we consider depends on the bid strategies at equilibrium (this is in the same spirit as the recent analysis of Feldman et al. \cite{FFGL13}) and, more interestingly, it is {\em deterministic}. In particular, it is defined as the product of the (expected) resource share a bidder receives in the optimal allocation and the expectation of bids of the other players at equilibrium.

Budget constraints are well-motivated in auction settings. In a slightly different context than ours, the effective welfare benchmark is considered by Dobzinski and Paes Leme, who call it {\em liquid} welfare in \cite{DPL13}. In proportional allocation, Syrgkanis and Tardos \cite{ST13} prove that the social welfare at equilibrium is a constant fraction of the optimal effective welfare. Note that our guarantee is considerably stronger as we compare directly the effective welfare at equilibrium with its optimal value.

\paragraph{Roadmap.} The rest of the paper is structured as follows. We begin with preliminary definitions in Section~\ref{sec:prelim}. Our price of anarchy bounds in terms of the social welfare are proved in Section~\ref{sec:main}. There, we also argue that in order to improve our analysis, radically new ideas are required. The budget-constrained setting is studied in Section~\ref{sec:ext}. We remark that we have not mentioned mixed Bayes-Nash equilibria in the above presentation of our results. Actually, we have observed that such equilibria coincide with pure ones even in the budget-constrained setting. We discuss related issues as well as additional open problems in Section~\ref{sec:open}.

\section{Preliminaries}\label{sec:prelim}
Each player (henceforth called {\em bidder}) $i$ in a proportional allocation game has a concave\footnote{Very recently, Correa et al.~\cite{CSS13} studied proportional allocation games in the less standard scenario of non-concave valuation functions.} non-decreasing valuation function $v_i:[0,1]\rightarrow \R^+$. A strategy for bidder $i$ is simply a non-negative bid. Given a bid vector $\bb=(b_1, b_2, ..., b_n)$, with one bid per bidder, the proportional allocation mechanism allocates to each bidder a fraction of the resource that is proportional to the bid submitted by her. Denoting by $d_i$ the resource share that is allocated to bidder $i$, it is $d_i=\frac{b_i}{\sum_j{b_j}}$. We often use the notation $B_{-i}$ to denote the sum of bids of all bidders besides $i$ (hence, $d_i=\frac{b_i}{b_i+B_{-i}}$). The utility of bidder $i$ from an allocation is simply the difference of her value for the fraction of the resource she gets minus her bid, i.e., $u_i(\bb)=v_i(d_i)-b_i$.

A bid vector $\bb$ is a pure equilibrium if the utility of all bidders is maximized, given the bid strategies of the other bidders. So, in a pure equilibrium, no bidder has any incentive to deviate to another strategy. Denoting by $(b'_i,\bb_{-i})$ the bid vector that is obtained from $\bb$ when bidder $i$ unilaterally deviates to bid strategy $b'_i$, we can express this condition as $u_i(\bb) \geq u_i(b'_i,\bb_{-i})$.

The social welfare of an allocation $d$ is the total value of bidders for the resource shares they receive, i.e., $SW(d) = \sum_i{v_i(d_i)}$. We denote by $SW^*$ the maximum value of the social welfare over all possible allocations. The price of anarchy over pure Nash equilibria is defined as the minimum value of the social welfare among all pure Nash equilibria divided by the optimal social welfare.

The bid strategy of a bidder $i$ can be randomized. In this case, $b_i$ is a random variable and the bidder aims to maximize her expected utility $\E{u_i(\bb)}$. The bid strategies of different bidders can be independent or correlated. A vector of independent randomized bid strategies is called a mixed Nash equilibrium if it simultaneously maximizes the expected utility of each bidder, given the bid strategies of the other bidders. More generally, coarse-correlated equilibria are solution concepts that capture correlated bid strategies. A vector of (possibly correlated) bid strategies is called a coarse-correlated equilibrium if no bidder has any incentive to unilaterally deviate to any deterministic bid strategy in order to improve her expected utility (again, given the strategies of the other bidders). The notion of the price of anarchy naturally extends to these solution concepts as well. For example, the price of anarchy over correlated equilibria is defined as the minimum value of the expected social welfare among all coarse-correlated equilibria divided by the optimal social welfare.

The above setting is known as the full (or complete) information setting. We consider the incomplete information (or Bayesian) setting as well; in this case, the valuation function $\vv_i$ of each bidder $i$ is drawn randomly (and independently from the other bidders) from a probability distribution $\F_i$ over concave, non-decreasing, and non-negative functions in $[0,1]$. Again, bidder $i$ aims to maximize her expected utility for each possible valuation function $v_i$ drawn from $\F_i$. In the incomplete information setting, each bidder $i$ bases her decision on her exact valuation $v_i$ and on the probability distributions according to which other bidders draw their valuations (and their corresponding bid strategies); these distributions are common knowledge.

So, the bid strategy of bidder $i$ is a (possibly random) bid function $b_i(\vv_i)$. A vector with one such strategy per bidder (with independence between bid strategies of different bidders) is called a mixed Bayes-Nash equilibrium if no bidder has any incentive to deviate to some other bid for any valuation function drawn from $\F_i$. In pure Bayes-Nash equilibria, bidders use deterministic bid functions. The price of anarchy over Bayes-Nash equilibria is defined as the minimum value of the expected social welfare among all Bayes-Nash equilibria divided by the expectation of the optimal social welfare. With some abuse in notation, we also use $SW^*$ to denote the expectation of the optimal social welfare in the Bayesian setting.

We also extend the above model by adding budget constraints to the bidders. In this setting, each bidder $i$ has a non-negative budget $c_i$ and she is never allowed to bid above her budget. This restriction can result to equilibria that have extremely low social welfare compared to the optimal one (whose definition does not take budgets into account). Following \cite{ST13} and \cite{DPL13}, we use the effective welfare benchmark in order to assess the quality of equilibria with budget-constrained bidders. The effective welfare of a (deterministic) allocation $d=(d_1, d_2, ..., d_n)$ is defined as $EW(d) = \sum_i{\min\{v_i(d_i),c_i\}}$. Note that the definition is similar to the definition of the social welfare; the important difference is that the value of each bidder is capped by her budget. We extend this definition to random allocations $d$ as $EW(d) = \sum_i{\min\{\E{v_i(d_i)},c_i\}}$. We denote by $EW^*$ the maximum value of the effective welfare over all allocations. The price of anarchy with respect to the effective welfare benchmark (over equilibria in a given class) is the minimum value of the effective welfare (among all allocations induced by equilibria in the class) divided by the optimal effective welfare.

In the Bayesian setting, both the budget $\cc_i$ of bidder $i$ and her valuation $\vv_i$ are drawn randomly according to the probability distribution $\F_i$. We refine the effective welfare benchmark in this case as $$EW(d) = \sum_i{\mathbb{E}_{(\vv_i,\cc_i)\sim \F_i}\left[\min\{\mathbb{E}_{(\vv_{-i},\cc_{-i})\sim \F_{-i}}\left[\vv_i(d_i)\right],\cc_i\}\right]},$$
where the inner expectation is taken over the valuation-budget value pairs of the other bidders once the pair for bidder $i$ has been fixed (and over the corresponding bid strategies). In order to simplify notation in the proofs below, we will not explicitly use the subscripts in the expectations.

%%%%%%%%%%%%%%%%%%%%%%%%%%%%%%%%%%%%%%%%%%%%%%%%%%%%%%%%%%%%%%%%

\section{Bounding the social welfare of equilibria}\label{sec:main}
In this section, we prove the price of anarchy bounds with respect to the social welfare. We consider both coarse-correlated equilibria in the full information setting as well as pure Bayes-Nash equilibria in the Bayesian setting. Our proofs use the following lemma which bounds the utility of a bidder at a deterministic deviation. We also use this lemma later in Section~\ref{sec:ext} where we study budget-constrained bidders.

\begin{lemma}\label{lem:deviation}
Consider a bidder with a concave and non-decreasing valuation function $v:[0,1]\rightarrow \R^+$ and let $\Gamma$ be the random variable denoting the sum of bids of the other bidders. Then, for every $z\in [0,1]$ and for every $\mu>0$, the expected utility the bidder would have by deviating to the deterministic bid $\mu z \E{\Gamma}$ is at least $\frac{3\mu-1}{4\mu}v(z)-\mu z \E{\Gamma}$.
\end{lemma}

\begin{proof}
It suffices to show that the expected value of the bidder when she deviates to the deterministic bid $y=\mu z \E{\Gamma}$ is at least $\frac{3\mu-1}{4\mu}v(z)$. Define the event $T:=\left\lbrace \Gamma \geq y \left( \frac{1}{z}-1 \right) \right\rbrace$. When $T$ is false, we have $\frac{y}{y+\Gamma} > z$ and, since $v$ is non-decreasing, we clearly have that
\begin{eqnarray*}
v\left(\frac{y}{y+\Gamma}\right)& \geq & v(z).
\end{eqnarray*}
Otherwise, when $T$ is true, $\frac{y}{y+\Gamma}\in [0,z]$. Since $v$ is concave and non-negative, its value in $[0,z]$ is lower-bounded by the line connecting points $(0,0)$ and $(z,v(z))$. Hence,
\begin{eqnarray*}
v\left(\frac{y}{y+\Gamma}\right)& \geq & \frac{y}{y+\Gamma}\cdot \frac{v(z)}{z}.
\end{eqnarray*}

So, we can bound the expected value of the bidder when she deviates to the deterministic bid $y$ using the two observation above and linearity of expectation.
\begin{eqnarray}\nonumber
\E{v \left( \frac{y}{y+\Gamma} \right)} &=& \E{ v \left( \frac{y}{y+\Gamma} \right) \one \overline{T} } + \E{v \left( \frac{y}{y+\Gamma} \right) \one T}\\\nonumber
&\geq & \E{v(z)\one{\overline{T}}}+\E{\frac{y}{y+\Gamma}\cdot \frac{v(z)}{z}\one{T}}\\\label{eq:global}
&=& v(z) (1-\Pr[T])+\frac{v(z)}{z} \E{\frac{y}{y+\Gamma}\one{T}}.
\end{eqnarray}
Here, we have used the notation $X\one{T}$ to denote the random variable that is equal to $X$ if $T$ is true and is zero otherwise.

We will now work with the rightmost term of the above right-hand side expression. Since the function $\frac{y}{y+\Gamma}$ is convex with respect to $\Gamma$, we can apply Jensen's inequality to obtain that
\begin{eqnarray*}
\E{\frac{y}{y+\Gamma}\one{T}} &=& \E{\frac{y}{y+\Gamma}|T}\cdot \Pr[T] \geq \frac{y\Pr[T]}{y+\E{\Gamma|T}}
\end{eqnarray*}
and, since $\E{X|T} \leq \frac{\E{X}}{\Pr[T]}$ for every random variable $X$, we have
\begin{eqnarray*}
\E{\frac{y}{y+\Gamma}\one{T}} &\geq & \frac{y\Pr[T]^2}{y\Pr[T]+\E{\Gamma}} \geq \frac{y\Pr[T]^2}{y+\E{\Gamma}}.
\end{eqnarray*}
The second inequality follows trivially since $\Pr[T]\leq 1$. Substituting $y$ and using the fact that $z\leq 1$, we get
\begin{eqnarray}\label{eq:sec-term}
\E{\frac{y}{y+\Gamma}\one{T}} &\geq & \frac{\mu z \Pr[T]^2}{1+ \mu z}\geq \frac{\mu z}{\mu+1}\Pr[T]^2.
\end{eqnarray}
Now, using (\ref{eq:global}), (\ref{eq:sec-term}), and the fact that $1-\alpha+\frac{\mu }{\mu+1}\alpha^2 \geq \frac{3\mu-1}{4\mu}$ for every $\alpha$, we obtain that
\begin{eqnarray*}
\E{v \left( \frac{y}{y+\Gamma} \right)} &\geq & v(z) \left(1-\Pr[T]+\frac{\mu}{\mu+1}\Pr[T]^2\right) \geq \frac{3\mu-1}{4\mu}v(z),
\end{eqnarray*}
as desired.
\end{proof}

We are ready to prove our price of anarchy bounds. We begin with the case of coarse-correlated equilibria in the full information setting which is much simpler.
\begin{theorem}\label{thm:cce}
The price of anarchy of proportional allocation games over coarse-correlated equilibria is at least $1/2$.
\end{theorem}

\begin{proof}
Consider a full information proportional allocation game with $n$ bidders in which bidder $i$ has valuation function $v_i$ and denote by $x_i$ the resource fraction bidder $i$ gets in the optimal allocation. Let $\bb$ be a coarse-correlated equilibrium that induces a random allocation $d=(d_1, ..., d_n)$ and let $B=\sum_i{b_i}$ be the random variable denoting the sum of bids of all bidders, with $B_{-i}$ being the sum of bids of all bidders besides bidder $i$. Since $\bb$ is a coarse-correlated equilibrium, bidder $i$ has no incentive to deviate to any deterministic bid (including the deviating bid $x_i \E{B_{-i}}$). By applying Lemma \ref{lem:deviation} for bidder $i$ with $z=x_i$, $\mu=1$ and $\Gamma=B_{-i}$, we obtain that
\begin{eqnarray*}
\E{u_i(\bb)} &\geq & \E{u_i(x_i \E{B_{-i}}, \bb_{-i})}\geq \frac{1}{2}v_i(x_i) - x_i \E{B_{-i}}.
\end{eqnarray*}
Summing over all bidders and using the fact that $B_{-i} \leq B$ for every bidder $i$, we have
\begin{eqnarray}\label{eq:to-be-discussed}
\sum_i{\E{u_i(\bb)}} & \geq & \frac{1}{2}\sum_i{v_i(x_i)} - \sum_i{x_i \E{B_{-i}}}\\\nonumber
&\geq & \frac{1}{2}\sum_i{v_i(x_i)} - \sum_i{x_i \E{B}}\\\nonumber
&=& \frac{1}{2}SW^* - \E{B}.
\end{eqnarray}
The theorem follows by this inequality since the social welfare equals the sum of bidders' utilities plus their bids, i.e., $\E{SW(d)} = \sum_i{\E{u_i(\bb)}}+\E{B}$.
\end{proof}

The last step of the proof above begins with inequality (\ref{eq:to-be-discussed}). Essentially, this inequality has the form
\begin{eqnarray*}
\sum_i{\E{u_i(\bb)}} & \geq & \lambda SW^* - \mu\sum_i{x_i \E{B_{-i}}}.
\end{eqnarray*}
The price of anarchy bound of \cite{ST13} follows after first proving an inequality of this type and then concluding to a price of anarchy bound of $\frac{\lambda}{\max\{1,\mu\}}$. The smoothness arguments of \cite{ST13} lead to a version of this inequality with $\lambda=2-\sqrt{3}$ and $\mu=1$. Here, we have been able to improve the parameters to $\lambda=1/2$ and $\mu=1$. The next lemma demonstrates that these parameters cannot be improved further.
\begin{lemma}
For every $\epsilon>0$, there exists a proportional allocation game such that for every $\lambda,\mu$ satisfying
\begin{eqnarray}\label{eq:lm}
\sum_i{u_i(\bb)} \geq \lambda SW^* - \mu \sum_i{ x_i B_{-i}}
\end{eqnarray}
where $x_i$ is the resource fraction of bidder $i$ in the optimal allocation and $B_{-i}$ is the sum of bids of all bidders besides bidder $i$ at a (pure Nash) equilibrium, it holds that $\frac{\lambda}{\max\{1,\mu\}}\leq \frac{1}{2}+\epsilon$.
\end{lemma}

\begin{proof}
Consider the proportional allocation game with $n\geq 2$ bidders in which bidder $1$ has valuation $v_1(x)=x$ and bidder $i$ has valuation $v_i(x)=\frac{n-1}{2n-3} x$ for $i\geq 2$. We can show that the bids in the (unique) pure Nash equilibrium are $b_1=1/4$ and $b_i=\frac{1}{4(n-1)}$ for $i\geq 2$. Indeed, assuming that this is true for all bidders besides $i$, it can be verified that the strategy $y$ that maximizes the utility $u_i(y,b_{-i}) = v_i\left(\frac{y}{y+B_{-i}}\right)-y$ for bidder $i$ satisfies $y=b_i$. I.e., bidder $1$ gets half of the resource and the remaining bidders share the remaining resource equally. Hence, $u_1(\bb) = v_1(1/2)-1/4 = 1/4$ and $\sum_{i\not=1}{u_i(\bb)} = (n-1)\left(v_i\left(\frac{1}{2(n-1)}\right)-\frac{1}{4(n-1)}\right) = \frac{1}{4(2n-3)}$.

In the optimal allocation, the whole resource is allocated to bidder $1$, i.e., $SW^*=1$, $x_1=1$ and $x_i=0$ for $i\geq 2$. Hence, inequality (\ref{eq:lm}) becomes
\begin{eqnarray*}
\frac{1}{4}+\frac{1}{4(2n-3)} &\geq & \lambda - \frac{\mu}{4}
\end{eqnarray*}
which implies that $\lambda\leq \max\{1,\mu\}\left(\frac{1}{2}+\frac{1}{4(2n-3)}\right)$. The lemma follows by setting $n$ sufficiently large.
\end{proof}

The proof for Bayes-Nash equilibria follows the same general approach with that of Theorem \ref{thm:cce}.
\begin{theorem}\label{thm:bne}
The price of anarchy of proportional allocation games over pure Bayes-Nash equilibria is at least $1/2$.
\end{theorem}

\begin{proof}
Consider an incomplete information proportional allocation game in which the valuation function $\vv_i$ of bidder $i$ is drawn from the probability distribution $\F_i$, independently for each bidder. Let $x_i$ be the random variable denoting the resource fraction bidder $i$ gets in the optimal allocation. Let $\bb$ be a pure Bayes-Nash equilibrium and $B$ be the random variable denoting the sum of bids of all bidders; again, $B_{-i}$ denotes the sum of bids of all bidders besides bidder $i$. Since $\bb$ is a pure Bayes-Nash equilibrium, bidder $i$ has no incentive to deviate to any deterministic bid (including the deviating bid $\E{x_i|v_i} \E{B_{-i}|v_i}$) when the valuation drawn from probability distribution $\F_i$ is $v_i$. So, in all conditional expectations below, we simply write $v_i$ to denote the event that the valuation $\vv_i$ drawn from $\F_i$ is $v_i$. By applying Lemma \ref{lem:deviation} for bidder $i$ with $z=\E{x_i|v_i}$, $\mu=1$ and $\Gamma=B_{-i}$, we obtain that
\begin{eqnarray*}
\E{u_i(\bb)|v_i} &\geq & \E{u_i(\E{x_i|v_i} \E{B_{-i}|v_i}, \bb_{-i})|v_i}\\
&\geq& \frac{1}{2}v_i(\E{x_i|v_i}) - \E{x_i|v_i} \E{B_{-i}|v_i}\\
&\geq & \frac{1}{2}\E{v_i(x_i)|v_i} - \E{x_i|v_i}\E{B}.
\end{eqnarray*}
The second inequality follows by Jensen's inequality since the valuation function $v_i$ is concave and due to the fact that in a pure Bayes-Nash equilibrium, the bid of a bidder different than $i$ does not depend on the exact valuation of bidder $i$ and, hence, $\E{B_{-i}|v_i}=\E{B_{-i}}\leq \E{B}$. Considering all possible valuations for bidder $i$ that are drawn from probability distribution $\F_i$, we have that her unconditional expected utility is
\begin{eqnarray*}
\E{u_i(\bb)} & \geq & \frac{1}{2}\E{\vv_i(x_i)} - \E{x_i} \E{B}.
\end{eqnarray*}

Summing over all bidders and using the facts that $\sum_i{x_i}=1$ and $B_{-i} \leq B$ for every bidder $i$, we have
\begin{eqnarray*}
\sum_i{\E{u_i(\bb)}} & \geq & \frac{1}{2}\sum_i{\E{\vv_i(x_i)}} - \sum_i{x_i \E{B}} = \frac{1}{2}SW^* - \E{B}.
\end{eqnarray*}
The theorem follows by this inequality since, again, the social welfare equals the sum of expected bidders' utilities plus the total amount of bids.
\end{proof}

\section{Budget-constrained bidders}\label{sec:ext}
In this section, we consider budget-constrained bidders and prove a lower bound of approximately $36\%$ and an upper bound of $50\%$ on the price of anarchy in terms of the effective welfare benchmark. Here, we prove Theorem \ref{thm:cce-budgets} for Bayes-Nash equilibria only; the (simpler) proof for coarse-correlated equilibria appears in Appendix~\ref{app:thm:cce-budgets}. Our upper bound (Theorem \ref{thm:cce-budgets-neg}) applies even to pure Nash equilibria.

Before proceeding to the presentation of our bounds for budget-constrained bidders, we remark that minor modifications of the proofs in the previous section can show that the social welfare over equilibria with budget-constrained bidders is at least $1/2$ of the optimal effective welfare, improving a corresponding bound of $26.8\%$ from \cite{ST13}. The necessary modifications are as follows. First, we need to define the deviating bids in terms of the resource shares in the allocation that maximizes the effective welfare. Then, there is a subtle case where Lemma \ref{lem:deviation} cannot be used, namely when the deviating bid for a bidder exceeds her budget. Fortunately, the inequality provided by Lemma~\ref{lem:deviation} follows trivially in this case (actually, we use this argument in the proof below). By repeating the analysis in the proofs of Theorems \ref{thm:cce} and \ref{thm:bne}, we can conclude that the social welfare at equilibrium is at least $1/2$ of the social welfare of the allocation that maximizes the effective welfare. The bound then follows by observing that the effective welfare of this allocation is upper-bounded by its social welfare.

\begin{theorem}\label{thm:cce-budgets}
The price of anarchy of proportional allocation games with budget-constrained bidders over coarse-correlated or Bayes-Nash equilibria is at least $0.3596$.
\end{theorem}

\begin{proof}
Let $\mu\in (1/3,1]$ be a parameter whose exact value will be defined later. Consider an incomplete information proportional allocation game with $n$ bidders in which the valuation function $\vv_i$ and the budget $\cc_i$ of bidder $i$ are drawn from the probability distribution $\F_i$, independently for each bidder. Let $x_i$ be the random variable denoting the resource fraction bidder $i$ gets in the allocation that maximizes the effective welfare. Let $\bb$ be a pure Bayes-Nash equilibrium that induces a random allocation $d=(d_1, ..., d_n)$ and $B$ be the random variable denoting the sum of bids of all bidders; again, $B_{-i}$ denotes the sum of bids of all bidders besides bidder $i$. We denote by $A_i$ the set that contains all pairs of a valuation function and a corresponding budget value $(v_i,c_i)$ that are drawn from the probability distribution $\F_i$ and satisfy $\E{\vv_i(d_i)|v_i} \leq c_i$. Consider a bidder $i$ with valuation-budget pair $(v_i,c_i)\not\in A_i$. By the definition of $A_i$, we have
\begin{eqnarray*}
\min\{\E{\vv_i(d_i)|v_i},c_i\} &\geq & \min\{\E{\vv_i(x_i)|v_i},c_i\}.
\end{eqnarray*}
By considering all valuation-budget pairs not belonging to $A_i$, we obtain
\begin{eqnarray*}
\E{\min\{\E{\vv_i(d_i)},\cc_i\}\one{(\vv_i,\cc_i)\not\in A_i}} &\geq & \E{\min\{\E{\vv_i(x_i)},\cc_i\}\one{(\vv_i,\cc_i)\not\in A_i}},
\end{eqnarray*}
and summing over all bidders, we have
\begin{eqnarray}\label{eq:first}
\sum_i{\E{\min\{\E{\vv_i(d_i)},\cc_i\}\one{(\vv_i,\cc_i)\not\in A_i}}} &\geq & \sum_i{\E{\min\{\E{\vv_i(x_i)},\cc_i\}\one{(\vv_i,\cc_i)\not\in A_i}}}.
\end{eqnarray}

Now consider a valuation-budget pair $(v_i,c_i)\in A_i$ for bidder $i$ that is drawn from $\F_i$. If $\mu\E{x_i|v_i}\E{B_{-i}|v_i}\leq c_i$, we can bound the expected utility $\E{u_i(\bb)|v_i}$ by considering the deviation of bidder $i$ to bid $\mu\E{x_i|v_i}\E{B_{-i}|v_i}$ (which is within bidder $i$'s budget $c_i$). By Lemma \ref{lem:deviation}, we have
\begin{eqnarray*}
\E{u_i(\bb)|v_i} &\geq & \frac{3\mu-1}{4\mu} v_i(\E{x_i|v_i}) - \mu \E{x_i|v_i}\E{B_{-i}|v_i}\\
&\geq & \frac{3\mu-1}{4\mu} \E{v_i(x_i)|v_i} - \mu \E{x_i|v_i}\E{B}\\
&\geq & \frac{3\mu-1}{4\mu} \min\{\E{v_i(x_i)|v_i},c_i\} - \mu \E{x_i|v_i}\E{B}.
\end{eqnarray*}
The second inequality follows by Jensen's inequality and by the fact $\E{B_{-i}|v_i}=\E{B_{-i}}$. Otherwise, if $\mu\E{x_i|v_i}\E{B_{-i}|v_i}> c_i$, the same inequality follows easily since
\begin{eqnarray*}
\E{u_i(\bb)|v_i} &\geq & 0\\
& > & c_i - \mu \E{x_i|v_i}\E{B_{-i}|v_i}\\
&\geq & \frac{3\mu-1}{4\mu} \min\{\E{v_i(x_i)|v_i},c_i\} - \mu \E{x_i|v_i}\E{B}.
\end{eqnarray*}
Hence, when $(v_i,c_i)\in A_i$, we have
\begin{eqnarray*}
\E{u_i(\bb)|v_i} + \mu \E{x_i|v_i}\E{B} &\geq & \frac{3\mu-1}{4\mu} \min\{\E{v_i(x_i)|v_i},c_i\}.
\end{eqnarray*}
By considering all valuation-budget values belonging to $A_i$, we have
\begin{eqnarray*}
\E{u_i(\bb)\one{(\vv_i,\cc_i)\in A_i}} + \mu \E{x_i\one{(\vv_i,\cc_i)\in A_i}}\E{B} &\geq & \frac{3\mu-1}{4\mu} \E{\min\{\E{\vv_i(x_i)},\cc_i\}\one{(\vv_i,\cc_i)\in A_i}}.
\end{eqnarray*}
Using the obvious fact that $\E{x_i}\geq \E{x_i\one{(\vv_i,\cc_i)\in A_i}}$ and the above inequality, we obtain that
\begin{eqnarray}\nonumber
\E{u_i(\bb)\one{(\vv_i,\cc_i)\in A_i}} + \mu \E{x_i}\E{B} &\geq &\E{u_i(\bb)\one{(\vv_i,\cc_i)\in A_i}} + \mu \E{x_i\one{(\vv_i,\cc_i)\in A_i}}\E{B}\\\label{eq:X}
&\geq & \frac{3\mu-1}{4\mu} \E{\min\{\E{\vv_i(x_i)},c_i\}\one{(\vv_i,\cc_i)\in A_i}}.
\end{eqnarray}
Now, we have
\begin{eqnarray}\nonumber
& & \sum_i{\E{\min\{\E{\vv_i(d_i)},\cc_i\}\one{(\vv_i,\cc_i)\in A_i}}}+\mu \sum_i{\E{\min\{\E{\vv_i(d_i)},\cc_i\}\one{(\vv_i,\cc_i)\not\in A_i}}}\\ \nonumber
&\geq & \sum_i{\E{u_i(\bb)+b_i\one{(\vv_i,\cc_i)\in A_i}}} + \mu\sum_i{\E{b_i\one{(\vv_i,\cc_i)\not\in A_i}}}\\\nonumber
&\geq & \sum_i{\E{u_i(\bb)+\mu b_i\one{(\vv_i,\cc_i)\in A_i}}} + \mu\sum_i{\E{b_i\one{(\vv_i,\cc_i)\not\in A_i}}}\\\nonumber
&=& \sum_i{\E{u_i(\bb)\one{(\vv_i,\cc_i)\in A_i}}} + \mu \E{B}\\\nonumber
&=& \sum_i{\left(\E{u_i(\bb)\one{(\vv_i,\cc_i)\in A_i}}+\mu \E{x_i}\E{B}\right)}\\\label{eq:second}
&\geq & \frac{3\mu-1}{4\mu} \sum_i{\E{\min\{\E{\vv_i(x_i)},\cc_i\}\one{(\vv_i,\cc_i)\in A_i}}}.
\end{eqnarray}
The first inequality follows since the quantity $\min\{\E{\vv_i(d_i)},\cc_i\}$ equals $\E{\vv_i(d_i)}$ when $(\vv_i,\cc_i)\in A_i$ and $\cc_i$ otherwise; in the latter case, the budget is clearly not smaller than the bid of bidder $i$. The second inequality follows since $\mu\leq 1$, the two equalities are obvious, and the last inequality follows by (\ref{eq:X}).
Now, using (\ref{eq:first}) and (\ref{eq:second}), we have
\begin{eqnarray*}
EW(d) &=& \sum_i{\E{\min\{\E{\vv_i(d_i)},\cc_i\}}}\\
& = & \sum_i{\E{\min\{\E{\vv_i(d_i)},\cc_i\}\one{(\vv_i,\cc_i)\in A_i}}}+\mu \sum_i{\E{\min\{\E{\vv_i(d_i)},\cc_i\}\one{(\vv_i,\cc_i)\not\in A_i}}}\\
& & +(1-\mu) \sum_i{\E{\min\{\E{\vv_i(d_i)},\cc_i\}\one{(\vv_i,\cc_i)\not\in A_i}}}\\
&\geq & \frac{3\mu-1}{4\mu} \sum_i{\E{\min\{\E{\vv_i(x_i)},\cc_i\}\one{(\vv_i,\cc_i)\in A_i}}}\\
& & +(1-\mu)\sum_i{\E{\min\{\E{\vv_i(x_i)},\cc_i\}\one{(\vv_i,\cc_i)\not\in A_i}}}\\
&\geq& \min\left\{\frac{3\mu-1}{4\mu},1-\mu\right\}\sum_i{\E{\min\{\E{\vv_i(x_i)},\cc_i\}}}\\
&=&  \min\left\{\frac{3\mu-1}{4\mu},1-\mu\right\} EW^*.
\end{eqnarray*}
Hence, the price of anarchy with respect to the effective welfare benchmark is bounded by the quantity $\min\left\{\frac{3\mu-1}{4\mu},1-\mu\right\}$ which is maximized to $\frac{7-\sqrt{17}}{8}\approx 0.3596$ for $\mu=\frac{1+\sqrt{17}}{8}$.
\end{proof}

We conclude this section by presenting our upper bound on the price of anarchy; note that it holds even for pure Nash equilibria.

\begin{theorem}\label{thm:cce-budgets-neg}
For every $\epsilon>0$, there exists a proportional allocation game among budget-constrained bidders with price of anarchy at most $1/2+\epsilon$ over pure Nash equilibria, with respect to the effective welfare benchmark.
\end{theorem}

\begin{proof}
Let $\alpha\in (0,1)$. Consider a proportional allocation game with two bidders. Bidder $1$ has valuation $v_1(x)=x$ and budget $c_1=\frac{\alpha}{(1+\alpha)^2}$. Bidder $2$ has valuation $v_2(x)=\alpha x$ and infinite budget. The state in which bidder $1$ bids $\frac{\alpha}{(1+\alpha)^2}$ (i.e., her budget) and bidder $2$ bids $\frac{\alpha^2}{(1+\alpha)^2}$ is a pure Nash equilibrium, since the derivatives $\frac{b_2}{(b_1+b_2)^2}-1$ and $\frac{\alpha b_1}{(b_1+b_2)^2}-1$ of the utilities of the bidders (as functions of their strategies) are equal to zero. Observe that $v_1\left(\frac{1}{1+\alpha}\right)$ significantly exceeds her budget for every value of $\alpha$. Hence, $EW(\bb) = \frac{\alpha}{(1+\alpha)^2}+\frac{\alpha^2}{1+\alpha} = \frac{\alpha+\alpha^2+\alpha^3}{(1+\alpha)^2}$. The optimal effective welfare is bounded by the welfare at the state when bidder $1$ bids her budget $\frac{\alpha}{(1+\alpha)^2}$ and bidder $2$ bids $1-\frac{\alpha}{(1+\alpha)^2}$ so that bidder $1$ gets value equal to her budget. Hence, the optimal effective welfare is $EW^* = \frac{2\alpha+\alpha^2+\alpha^3}{(1+\alpha)^2}$. Clearly, the ratio $EW(\bb)/EW^*$ approaches $1/2$ from above as $\alpha$ approaches $0$. The theorem follows by selecting $\alpha$ to be sufficiently small.
\end{proof}

\section{Discussion and open problems}\label{sec:open}
Our work leaves the obvious open problem of computing the tight bound on the price of anarchy over coarse-correlated and Bayes-Nash equilibria. So far, the only upper bound that is known is the counter-example of $3/4$ from \cite{JT04} for pure Nash equilibria. Is $3/4$ the tight bound for all equilibrium concepts considered in the current paper? Actually, we have not been able to identify any coarse-correlated equilibrium in the full information model that is non-pure. Do such equilibria really exist? Interestingly, we show in Lemma~\ref{lem:no-mixed} that mixed Nash equilibria coincide with pure ones. More generally, this statement applies to mixed Bayes-Nash equilibria in the budget-constrained setting (proof in Appendix~\ref{app:lem:no-mixed}). Does it extend to coarse-correlated ones? We believe that this is an interesting open problem.

%%%%%%%%%%%%%%%% theorem: No mixed equlibria %%%%%%%%%%%%%%%%%%%
\begin{lemma}\label{lem:no-mixed}
The set of mixed Bayes-Nash equilibria in any proportional allocation game (possibly with budget-constrained bidders) coincides with that of pure Bayes-Nash equilibria.
\end{lemma}

In the Bayesian setting, we have not considered more general equilibrium concepts such as coarse-correlated Bayesian equilibria. The main reason is that our analysis requires that the expectation of the sum of bids of the other bidders is the same for any possible valuation bidder $i$ can draw from her distribution; this property is not satisfied by more general equilibrium concepts. What is the price of anarchy in this case? Interestingly, the answer cannot be $3/4$ as our next counter-example indicates (proof in Appendix~\ref{app:lem:lb}).

\begin{lemma}\label{lem:lb}
There exists a proportional allocation game that has price of anarchy at most $0.7154$ over coarse-correlated Bayesian equilibria.
\end{lemma}

Also, recall that we have assumed that bidders have independent valuations. This is a typical assumption in the Bayes-Nash price of anarchy literature \cite{BR10,CKS08,FFGL13,KMST13,R12,S12,ST13} with \cite{CKK+12} being the only exception we are aware of. Unfortunately, our proof of the pure Bayes-Nash price of anarchy bound does not carry over to the case of correlated valuations either (for the same reason mentioned above). Still, we have not been able to find any counter-example with non-constant price of anarchy in this setting. Again, what is the price of anarchy in this case? These questions are interesting in the budget-constrained setting as well.

%\small

\normalsize
\newpage\appendix

\section{Proof of Theorem~\ref{thm:cce-budgets} for coarse-correlated equilibria}\label{app:thm:cce-budgets}
Let $\mu\in (1/3,1]$ be a parameter whose exact value will be defined later. Consider a full information proportional allocation game with $n$ bidders in which bidder $i$ has valuation function $v_i$ and budget $c_i$ and denote by $x_i$ the resource fraction bidder $i$ gets in the allocation that maximizes the effective welfare. Let $\bb$ be a coarse-correlated equilibrium inducing an allocation $d=(d_1, ..., d_n)$ and let $B=\sum_i{b_i}$ be the random variable denoting the sum of bids of all bidders, with $B_{-i}$ being the sum of bids of all bidders besides bidder $i$. Let $A$ be the set of bidders with $\E{v_i(d_i)} \leq c_i$. Clearly, for every bidder not belonging to set $A$, it holds that
\begin{eqnarray}\nonumber
\min\{\E{v_i(d_i)},c_i\} &\geq &\min\{v_i(x_i),c_i\}.
\end{eqnarray}
Summing over all bidders not belonging to $A$ (and multiplying by $1-\mu$), we obtain that
\begin{eqnarray}\label{eq:not-A}
(1-\mu) \sum_{i\not\in A}{\min\{\E{v_i(d_i)},c_i\}} &\geq & (1-\mu) \sum_{i\not\in A}{\min\{v_i(x_i),c_i\}}.
\end{eqnarray}

For every bidder $i\in A$, we distinguish between two cases. If $\mu x_i\E{B_{-i}} > c_i$, then clearly
\begin{eqnarray*}
\E{u_i(\bb)} &\geq & 0\\
&> &c_i - \mu x_i\E{B_{-i}}\\
&\geq & \frac{3\mu-1}{4\mu}\min\{\E{v_i(x_i)},c_i\} - \mu x_i\E{B_{-i}}.
\end{eqnarray*}
In order to prove the same inequality when $\mu x_i\E{B_{-i}} \leq c_i$, we bound the utility of bidder $i$ by the utility she would have when deviating to bid $\mu x_i\E{B_{-i}}$ (which is within $i$'s budget $c_i$). Using Lemma \ref{lem:deviation}, we have again
\begin{eqnarray*}
\E{u_i(\bb)} &\geq & \frac{3\mu-1}{4\mu}\E{v_i(x_i)} - \mu x_i\E{B_{-i}}.
\end{eqnarray*}
Summing this last inequality over all bidders of $A$, we obtain
\begin{eqnarray*}
\sum_{i\in A}{\E{u_i(\bb)}} &\geq & \frac{3\mu-1}{4\mu}\sum_{i\in A}{\min\{\E{v_i(x_i)},c_i\}} - \sum_{i\in A}{\mu x_i\E{B_{-i}}}\\
& \geq & \frac{3\mu-1}{4\mu}\sum_{i\in A}{\min\{\E{v_i(x_i)},c_i\}} - \mu \E{B}\sum_{i\in A}{x_i}\\
& \geq & \frac{3\mu-1}{4\mu}\sum_{i\in A}{\min\{\E{v_i(x_i)},c_i\}} - \mu \E{B}.
\end{eqnarray*}
Using the equality $B=\sum_i{b_i}$ and linearity of expectation, this inequality implies that
\begin{eqnarray*}
\sum_{i\in A}{\left(\E{u_i(\bb)}+\mu \E{b_i}\right)} +\mu \sum_{i\not\in A}{\E{b_i}} & \geq & \frac{3\mu-1}{4\mu}\sum_{i\in A}{\min\{\E{v_i(x_i)},c_i\}}.
\end{eqnarray*}
Since $\E{u_i(\bb)}+\mu \E{b_i} \leq \min\{\E{v_i(d_i)},c_i\}$ for every bidder in $A$ (recall that $\mu\leq 1$ and $v_i(d_i)=u_i(\bb)+b_i$) and $\E{b_i}\leq \min\{\E{v_i(d_i)},c_i\}$ for every bidder not belonging to $A$, the above inequality yields
\begin{eqnarray}\label{eq:A}
\sum_{i\in A}{\min\{\E{v_i(d_i)},c_i\}} +\mu \sum_{i\not\in A}{\min\{\E{v_i(d_i)},c_i\}}& \geq & \frac{3\mu-1}{4\mu}\sum_{i\in A}{\min\{\E{v_i(x_i)},c_i\}}
\end{eqnarray}
By summing (\ref{eq:not-A}) and (\ref{eq:A}), we obtain
\begin{eqnarray*}
EW(d) &=& \sum_{i\in A}{\min\{\E{v_i(d_i)},c_i\}} +\sum_{i\not\in A}{\min\{\E{v_i(d_i)},c_i\}}\\
& \geq & \frac{3\mu-1}{4\mu}\sum_{i\in A}{\min\{\E{v_i(x_i)},c_i\}}+(1-\mu) \sum_{i\not\in A}{\min\{\E{v_i(x_i)},c_i\}}\\
&\geq & \min\left\{\frac{3\mu-1}{4\mu},1-\mu\right\}\sum_{i}{\min\{\E{v_i(x_i)},c_i\}}\\
&=& \min\left\{\frac{3\mu-1}{4\mu},1-\mu\right\} EW^*.
\end{eqnarray*}
Hence, the price of anarchy with respect to the effective welfare benchmark is bounded by the quantity $\min\left\{\frac{3\mu-1}{4\mu},1-\mu\right\}$ which is maximized to $\frac{7-\sqrt{17}}{8}\approx 0.3596$ for $\mu=\frac{1+\sqrt{17}}{8}$.
\qed

\section{Proof of Lemma~\ref{lem:no-mixed}}\label{app:lem:no-mixed}
Assume that there exists an incomplete information proportional allocation game (possibly with budget-constrained bidders) that has a mixed Bayes-Nash equilibrium $\bb$ in which bidder $i$ bids two different values $y_1$ and $y_2$ (with $y_1<y_2$) with non-zero probability when her valuation function is $v_i$; both values are within the budget of bidder $i$ (if any). We will show that this is not possible.

By the mixed Bayes-Nash equilibrium condition, both $y_1$ and $y_2$ should yield the same maximum expected utility $U$ to bidder $i$, i.e.,
\begin{eqnarray*}
U = \E{u_i(y_1,b_{-i})|v_i} = \E{u_i(y_2,b_{-i})|v_i}.
\end{eqnarray*}
Let $B_{-i}$ be the random variable denoting the sum of bids of all bidders besides bidder $i$ and let $f(y)=\E{v_i \left( \frac{y}{y+B_{-i}} \right)}$ be the expected value of bidder $i$ when unilaterally deviating to bid $y$. Clearly, $f$ is non-decreasing. It is also concave in $[y_1,y_2]$ since it is defined as the linear combination of concave functions: for every value of $B_{-i}$, $v_i \left( \frac{y}{y+B_{-i}} \right)$ is a concave function with respect to $y$, and the expectation over $B_{-i}$ is simply a linear combination over such functions. Clearly, $\E{u_i(y,b_{-i})}=f(y)-y$.

We furthermore claim that $f$ is strictly increasing in $[y_1,y_2]$. If this was not the case, then due to the concavity of $f$ there should exist $y' \in [y_1,y_2)$ such that $f(y')=f(y_2)$ and, hence, bidder $i$ could deviate to bid $y'$ (which is clearly within her budget, if any) for an improved expected utility of $\E{u_i(y',b_{-i})|v_i}=f(y')-y'>f(y_2)-y_2=U$. This would contradict the mixed Bayes-Nash equilibrium condition.

The fact that $f$ is strictly increasing clearly implies that $\Pr[B_{-i}>0]>0$. But then, for every positive value of $B_{-i}$, $v_i \left( \frac{y}{y+B_{-i}} \right)$ is a strictly concave function of $y$ and, subsequently, $f$ is also strictly concave as a linear combination of concave functions including strictly concave ones. Hence, there exists $\lambda \in [0,1]$ such that $f(\lambda y_1 + (1-\lambda)y_2)> \lambda f(y_1) + (1-\lambda)f(y_2)$ and bidder $i$ has a profitable deviation to bid $y'=\lambda y_1 + (1-\lambda)y_2$ (which is again within her budget, if any) since
\begin{eqnarray*}
\E{u_i(y',b_{-i})} & = & f(y')-y' \\
& > & \lambda f(y_1) + (1-\lambda)f(y_2) -\lambda y_1 - (1-\lambda)y_2 \\
& = & \lambda(f(y_1)-y_1)+(1-\lambda)(f(y_2)-y_2) \\
& = & U.
\end{eqnarray*}
We conclude that the support of any mixed Bayes-Nash equilibrium cannot contain two different bid values for bidder $i$ when her valuation is $v_i$ and, subsequently (by extending the same argument to all possible valuations of bidder $i$ and to all bidders), it must be a pure Bayes-Nash equilibrium.
\qed

We remark that if the valuation functions are differentiable (we do not make any such assumption in the above proof), a much simpler proof follows by observing that the utility of bidder $i$, when seen as a function of bidder $i$'s strategy, has strictly decreasing derivative. Then, the utility is maximized either by a bid equal to the budget of the bidder (if any) or at the unique bid that nullifies its derivative.

\section{Sketch of proof of Lemma~\ref{lem:lb}}\label{app:lem:lb}
A vector of possibly correlated bid functions is called a coarse-correlated Bayesian equilibrium if no bidder $i$ has any incentive to unilaterally deviate to any deterministic bid strategy in order to improve her expected utility (again, given the strategies of the other bidders), for any valuation she draws from her probability distribution $\F_i$. Coarse-correlated Bayesian equilibria are more general than mixed Bayes-Nash equilibria since the bid functions of different bidders are not restricted to being independent.

Our counter-example has two bidders. Bidder 1 has valuation function $x$ with probability $p_1$ and $\epsilon_1 x$ with probability $1-p_1$. Bidder 2 has valuation function $\alpha x$ with probability $p_2$ and $\epsilon_2 x$ with probability $1-p_2$. We require $1\geq\alpha$ and, furthermore, $\alpha$ is significantly larger than $\epsilon_1$ which in turn is significantly larger than $\epsilon_2$.

We construct a coarse-correlated Bayesian equilibrium of the following form: When the valuations of the bidders are $x$ and $\alpha x$, the bid strategies are $\gamma$ and $\delta$ respectively. These values are significant and yield constant resource fractions to both bidders. In all other cases where at least one of the bidders has an almost zero valuation, the bids are extremely close to zero. However, the bidder that has significantly higher valuation than the other submits a significantly higher bid (but still very close to zero) and gets almost $100\%$ of the resource. In the following, we round negligibly small bids or valuations to $0$ and treat an allocation of almost $100\%$ of the resource to some bidder as exactly $100\%$. This rounding does not affect the final result that we can obtain but a significantly more detailed (and tedious) calculations are needed for a formal proof. So, we will assume that when the valuations are $x$ or $\epsilon_1 x$ for bidder 1 and $\epsilon_2 x$ for bidder 2, the bid strategies are (almost) $0$ but the bid of bidder 1 is significantly higher so that she gets (almost) $100\%$ of the resource. Similarly, when the valuations are $\epsilon_1 x$ and $\alpha x$, the bid strategies are (almost) $0$ but the bid of bidder 2 is significantly higher so that she gets almost $100\%$ of the resource.

Notice that the expected utility of bidder 1 when her valuation is $x$ is (approximately) $p_2\frac{\gamma}{\gamma+\delta}+1-p_2-p_2 \gamma$ and becomes (approximately) $p_2\frac{y}{y+\delta}+1-p_2-y$ when deviating to a deterministic bid $y$. We require that the first quantity is higher than the second one so that no such deviation exists, i.e., $p_2\frac{\gamma}{\gamma+\delta}-p_2 \gamma\geq p_2\frac{y}{y+\delta}-y$ for every $y\geq 0$. Similarly, we require that $p_1\frac{\alpha \delta}{\gamma+\delta}-p_1 \delta\geq p_1\frac{y}{\gamma+y}-y$. Note that the right-hand side of the above constraints are maximized to $(\sqrt{p_2}-\sqrt{\delta})^2$ and $(\sqrt{\alpha p_1}-\sqrt{\gamma})^2$, respectively. It remains to compute the exact bid values that satisfy these constraints and minimize the price of anarchy. This is done in the following non-linear mathematical program
\begin{eqnarray*}
\mbox{minimize} & & \frac{p_1p_2\left(\frac{\gamma}{\gamma+\delta}+\alpha\frac{\delta}{\gamma+\delta}\right)+p_1(1-p_2)+\alpha(1-p_1)p_2}{p_1+\alpha (1-p_1)p_2}\\
\mbox{subject to:} & & p_2 \frac{\gamma}{\gamma+\delta}-p_2 \gamma \geq (\sqrt{p_2}-\sqrt{\delta})^2\\
& & \alpha p_1\frac{\delta}{\gamma+\delta}-p_1 \delta \geq (\sqrt{\alpha p_1}-\sqrt{\gamma})^2\\
& & \gamma, \delta \geq 0, 0\leq \alpha \leq 1, 0\leq p_1, p_2 \leq 1
\end{eqnarray*}
which has been solved using Matlab to give an upper bound of $0.7154$ for $\alpha=0.2913$, $\gamma=0.1071$, $\delta=0.1510$, $p_1=0.6682$, and $p_2=0.7616$.\qed
\end{document}